\documentclass[12pt]{article}

\usepackage{amssymb,amsmath,amsfonts,eurosym,geometry,ulem,graphicx,caption,color,setspace,sectsty,comment,footmisc,caption,natbib,pdflscape,subfigure,array,hyperref,tikz}
\usepackage{natbib}
\usepackage{url}
\usepackage{todonotes}

\newtheorem{theorem}{Theorem}

\newtheorem{proposition}{Proposition}
\newtheorem{lemma}[theorem]{Lemma}
\newenvironment{proof}[1][Proof]{\noindent\textbf{#1.} }{\ \rule{0.5em}{0.5em}}

\newcolumntype{L}[1]{>{\raggedright\let\newline\\arraybackslash\hspace{0pt}}m{#1}}
\newcolumntype{C}[1]{>{\centering\let\newline\\arraybackslash\hspace{0pt}}m{#1}}
\newcolumntype{R}[1]{>{\raggedleft\let\newline\\arraybackslash\hspace{0pt}}m{#1}}

\usepackage{natbib}

\title{``Zero Cost'' Majority Attacks on Permissionless Blockchains\thanks{The latest version of this paper is available at joshuagans.com. Thanks to Yackolley Amoussou-Guenou, Yannis Bakos, Joe Bonneau, Eric Budish, Guillaume Haeringer, Raveesh Mayya and seminar participants at the a16z Crypto Lab for useful discussions. All errors remain our own.}}

\author{Joshua S. Gans\thanks{corresponding author, Rotman School of Management, University of Toronto and NBER, email: joshua.gans@utoronto.ca} \and Hanna Halaburda\thanks{Stern School of Management, New York University, email: hhalaburda@gmail.com}}

\begin{document}

\maketitle

\begin{abstract}
The core premise of permissionless blockchains is their reliable and secure operation without the need to trust any individual agent. At the heart of blockchain consensus mechanisms is an explicit cost (whether work or stake) for participation in the network and the opportunity to add blocks to the blockchain. A key rationale for that cost is to make attacks on the network, which could be theoretically carried out if a majority of nodes were controlled by a single entity, too expensive to be worthwhile. We demonstrate that a majority attacker can successfully attack with a {\em negative cost}, which shows that the protocol mechanisms are insufficient to create a secure network, and emphasizes the importance of socially driven mechanisms external to the protocol. At the same time, negative cost enables a new type of majority attack that is more likely to elude external scrutiny.
   
\bigskip
\bigskip
   
\end{abstract}

\newpage

\section{Introduction}

The introduction of Bitcoin by \cite{nakamoto2008bitcoin} ushered in an era of permissionless blockchains. These blockchains promise reliable and secure operation without the need to trust any individual agent. Over a decade later, these blockchains have expanded beyond Bitcoin and now attract billions of dollars worth of economic activity, as well as new commercial developments. However, their security vulnerabilities are still not fully understood.

Nakamoto's innovative trust system relies on the blockchain protocol. However, these networks function as intended only if the majority of the controlling force (mining or staking) behaves ``honestly,'' i.e., follows the protocol. A deviating majority could influence the blockchain record -- double-spending, censoring transactions, or sabotaging the network. It can do so via a ``majority attack,'' which involves generating a separate blockchain of records called a ``fork.''

Therefore, at the heart of blockchain consensus mechanisms is an explicit cost (whether Proof of Work or Stake) for participation in the network and the opportunity to add blocks to the blockchain. In return, the participants are compensated with a mining reward, in the form of the block reward and transaction fees, when they succeed in adding a block. \cite{nakamoto2008bitcoin} argued that these incentives provided by the protocol would be enough to make even the majority miner behave honestly:

\begin{quote}
The incentive may help encourage nodes to stay honest. If a greedy attacker is able to
assemble more CPU power than all the honest nodes, he would have to choose between using it
to defraud people by stealing back his payments, or using it to generate new coins. He ought to
find it more profitable to play by the rules, such rules that favour him with more new coins than
everyone else combined, than to undermine the system and the validity of his own wealth.
(\cite{nakamoto2008bitcoin}, p.4)
\end{quote}

\noindent Since then, researchers (e.g., \cite{budish2022economic}, \cite{chiu2022economics}) have shown that the protocol may not be enough to assure security. Specifically, \cite{budish2022economic} shows that the cost of the majority attack may be zero. This is because an attacker would continue to earn mining rewards associated with confirming new blocks to what will eventually become the consensus chain. Such payments will mitigate the net cost of an attack. If the cost of mining is the same for the attackers as for the honest miners, the block rewards will cover all the mining costs (otherwise, honest miners would not participate), leading to zero net cost of an attack. That means that the majority miner has an incentive to double-spend whenever the value gained in the double-spend is positive. 

Indeed, we see some double-spending in Proof of Work blockchains (e.g., \cite{shanaev2019cryptocurrency}, \cite{moroz2020double}, \cite{ramos2021great}). But why don't we see more of these behaviors? Since blockchains operate on peer-to-peer networks with only local information, some attacks may succeed without being detected. However, external forces are probably contributing to the honest behavior of miners. These forces may be societally driven, like the price of cryptocurrency dropping in response to detected double spending, or individual, like waiting for six confirmations before considering a transaction executed. These type of external mechanisms in the context of Proof of Work was formalized by \cite{biais2019blockchain} and \cite{budish2022economic}, capturing the potential reduction in the value of cryptocurrencies to reduce the value the attacker can realize from the attack. \cite{saleh2021blockchain} and \cite{bakos2021tradeoffs} investigate a similar mechanism for Proof of Stake blockchains. This literature makes it clear that understanding where the protocol falls short internally points to the external forces needed to secure the blockchains.

Therefore, here we generalize and extend the analysis of the majority attack, focusing on internal mechanisms that mediate attack costs. We show that the net cost of an attack may be not only zero but {\em negative}. With negative cost, the attack is attractive even if there is no double spending -- this is because the attacker benefits from higher mining rewards during the attack. Such an effect is possible for a majority attack mounted by a miner active on the blockchain before the attack, even if it does not control the majority on that blockchain. We obtain that result by properly accounting for opportunity costs and transaction fees --- part of the mining reward omitted in the previous literature.\footnote{\cite{huberman2021monopoly} do examine the determination of fees but only under the assumption that no attacks are in progress.} 

As the attacker withdraws their mining power from the honest chain, the majority attack will succeed with less attacking power than applied to the blockchain before the attack. Thus, the capacity to process transactions will be reduced. This reduced capacity means that the fees per block processed by the attacker would rise; intuitively, they can pick and choose higher fee transactions to fill reduced processing capacity during the attack. This means the attacker would expect to earn higher payouts for confirming transactions during the attack than otherwise. We demonstrate that the net cost of an attack will be negative.

Such a negative cost result highlights a new type of majority attack. Previous results typically focused on double-spend attacks, where with zero cost of an attack, the majority would have the incentive to double-spend any positive amount. With a negative cost of attack, it is profitable for the majority miner to ``re-organize" the blockchain even with no benefit from double-spending. This is because the attacker captures additional mining rewards.

As mentioned earlier, blockchain security must rely on socially driven forces external to the protocol. By the nature of peer-to-peer networks, detecting misbehavior on blockchains is challenging. But this new type of majority attacks may be even harder to monitor -- it may be conducted in smaller spurts so that even honest nodes may not detect it. Another consequence of this new type of majority attack is that acquiring and commanding the majority of mining power turns out to be more profitable than what we previously believed, which may be the force leading to the concentration of mining. Moreover, these non-double-spending attacks make small miners unprofitable, as their blocks get orphaned. As those miners exit, concentration further increases, which goes against the spirit and premise of decentralized distributed blockchain systems.

While our main analysis focuses on Proof of Work, we show that the forces leading to the negative cost of attack are general to Nakamoto consensus (based on the longest chain rule) and also apply to Proof of Work mechanism. Overall, our result shows that careful examination of the economic vulnerabilities of blockchain protocols can point us to new areas needing societal scrutiny if this new technology is to fulfil its promise.

The paper proceeds as follows. The next section sets up the model by looking at the outcome in the absence of a majority attack on a permissionless Proof of Work blockchain such as Bitcoin. Section 3 describes the attack leading to the paper's main contribution, where the costs of a majority attack are examined in close detail. There we highlight the role of transaction fees. Section 5 then considers extensions. Section 6 then applies the same type of analysis to a class of Proof of Stake protocols, while a final section concludes.

\section{Benchmark Outcome Without an Attack}

We start by examining Proof of Work protocols. Under these protocols, miners operate nodes that devote resources (computation plus electricity) as part of a game to determine which will propose the next block on a blockchain.\footnote{Proof of Stake is examined in the appendix. For readers unfamiliar with these consensus mechanisms, see the over-views by, for example, \cite{halaburda2022microeconomics}, \cite{halaburda2022beyond} and \cite{gans2023consensus}.} Following \cite{ma2018market} and \cite{biais2019blockchain}, let $h_i$ denote the resource allocation by miner $i$ to a blockchain, where $i \in \mathcal{M}$, the set of potential miners. The cost of providing $h_i$ is $c_i(h_i)$, a non-decreasing, (weakly) quasi-convex function. Let $H$ denote the total hash power in any given period, i.e., $H=\sum_{i\in\mathcal{M}} h_i$, where we drop time subscripts for convenience. In this environment, if the time it takes miner $i$ to solve the required computational puzzle is $y_i$, then this time is a random variable with an exponential distribution with parameter $\frac{h_i}{D}$ where $D$ parameterizes the difficulty of the computational puzzle. The first miner to solve the puzzle, which happens at time $Y=\min\{y_1,y_2,...\}$, confirms a block. By the properties of exponential distributions, $Y$ also follows an exponential distribution --- with parameter $\frac{1}{D}H$. Difficulty ($D$) is adjusted periodically (in Bitcoin every two weeks) so that, on average, a block is mined every $\tau$ periods. That is, $\tau = \frac{1}{H/D}$ or $D=\tau\,H$. Initially, we assume that all agents take $D$ as fixed.\footnote{Another way of expressing this outcome is that, if a node contributes resources, then there is a probability, $p_i$, they will be selected as a leader to propose, and perhaps confirm, a block of transactions at a given point in time. \cite{leshno2020bitcoin} show that when the selection probability, $p_i$, is equal to $\frac{h_i}{H}$, then this satisfies properties such as anonymity, Sybil resistance and zero returns to merging. That same selection probability is what is explicitly used in Proof of Stake protocols and the implied probability in Proof of Work protocols.} 

A miner who solves the puzzle first and proposes a valid block receives a reward of newly minted tokens with the expected value of $R$ (stated in units of fiat currency, meaning that they are in the same units as mining costs). In addition, the miner receives transaction fees for transactions included in the block with an expected (fiat currency) value of $\Phi$.

All of this implies that miner $i$ expects to mine $\frac{h_i}{\tau\, H}$ blocks per unit of clock time. Thus, $i$'s expected continuation profits are:
$$\frac{h_i}{\tau\, H}(R+\Phi)-c_i(h_i)\, .$$
Let $h^*_i$ be the hash power that maximizes these expected profits. Then the \textit{participation constraint} for miner $i$ is
\begin{equation}\label{participation}
    \frac{h^*_i}{\tau\,H}(R+\Phi) \ge c_i(h^*_i) \, .
\end{equation}
Thus, $\mathcal{M}' \equiv \{i \in \mathcal{M}|\frac{h^*_i}{\tau\, H}(R+\Phi) \ge c_i(h^*_i)\}$ is the set of active miners in a period for given $\tau$, $R$ and $\Phi$. Note that 
while the cardinality of $\mathcal{M}'$ is weakly smaller than of $\mathcal{M}'$, but with all miners optimizing in the benchmark blockchain without attack,  $\sum_{i\in\mathcal{M}'}h^*_i=\sum_{i\in\mathcal{M}}h^*_i=H$.

While a miner can secretly add the next block to its local record, it will be able to spend its $R+ \Phi$ only when the system, including other miners, reaches a consensus that this block is included in the blockchain. The procedure by which an honest miner would do this would be to broadcast the block along with their winning answer to the computational contest. Other nodes would then confirm the validity of the block and add extend their own records accordingly. 
In this way, the new block is made public, and this also triggers the conditions for the next computational contest. That said, it is possible that there may be lags in the broadcast of new blocks, and so some miners will be unaware that the contest has been won. In this situation, it is possible that they could also find a solution to the contest and propose a new block. Thus, there might be two distinct chains that were forked at the previous point of consensus. To resolve this situation, \cite{nakamoto2008bitcoin} proposed a convention -- called the longest chain rule (or Nakamoto consensus) -- that instructed nodes when there is a fork to consider the blockchain with the most confirmed blocks (i.e., ``the greatest proof-of-work effort invested in it,'' \cite{nakamoto2008bitcoin}, p.3) as the correct one to extend.\footnote{\cite{biais2019blockchain} demonstrated that it is, in fact, an equilibrium for honest nodes to adhere to this convention and extend the longest chain.}

\section{Description of an Attack}

The longest chain rule convention can be exploited to implement a majority attack. The convention calls for miners who win the computational contest and propose blocks to broadcast the new block publicly, extending the current consensus chain. Miners are not compelled to do this and could instead keep the block secret and, in the process, at least for a time, be the only entity working to solve the next computational puzzle. Doing so, however, risks that other miners will solve both the previous and new computational puzzle quicker and hence, create a longer chain. Only an attacker who controls a majority of the hash power can be reasonably assured that by mining a chain privately they could end up with the longest chain at some point. They could then broadcast that chain and have it established as the new consensus. In the process, they may confirm blocks during the attack that are not the blocks that would be proposed by honest nodes -- ``re-writing'' what those observing the public chain consider a final record. 

An attack may utilize {\em inside} hash power, which was already used in the blockchain before the attack, and {\em outside} hash power deployed in addition to what was used before the attack.  \cite{budish2022economic} focuses on fully outside majority attacks. In a fully outside attack, holding~$D$ as constant for the duration of the attack, attacker $A$ adds more hash power ($h_A$) to the network than the aggregate hash power on the blockchain before the attack, i.e., $h_A>H$. In a fully inside majority attack, the attacker is someone who can coordinate or acquire control of $h_A$ incumbent nodes, which constitutes a majority of the incumbent hash power, i.e. $h_A > \frac{1}{2}H$. More generally, an attack can comprise a mixture of inside and outside hash power. That is, $h_A$ must exceed $\underline{h}_A$, the minimum hash power for a successful attack where $\underline{h}_A \equiv \max\{\frac{1}{2}H,H-h^*_A\}$. By~\eqref{participation}, $h^*_A$ is the hash power that is chosen by the attacker prior to the attack. For some cases below, it will be convenient to denote $\alpha$ as the share of hash power procured internally (i.e., $\alpha = \frac{h^*_A}{h_A}$). Then a majority attack requires $h_A > H-\alpha h_A$.

As noted above, a majority attack is an attempt to control and re-write transactions on the blockchains. We follow \cite{budish2022economic} in assuming that the (net presented discounted expected) private benefits to an attacker are $V_{attack}$. In specific applications, $V_{attack}$ can be derived from further analysis to calculate the returns to double-spending or the censorship of other users. 

The attack involves the majority miner creating a fork of the blockchain that they keep private. We assume, in our baseline case, that difficulty does not adjust during the attack.\footnote{We explore difficulty adjustment in Section~\ref{sec:difficulty}.} On the attacking chain, all the blocks are mined by the attacker; the expected number of its blocks per unit of time is $\frac{h_A}{D}=\frac{h_A}{\tau\,H}$. On the honest chain during the attack, the expected number of blocks per unit of time is $\frac{H - \alpha h_A}{\tau\, H}<\frac{h_A}{\tau\,H}$. The attacking chain is mining blocks faster on average than the honest one. Therefore, the attack will succeed with probability~1, and $\mathcal{L}$ denotes the expected duration of the attack.\footnote{\cite{budish2022economic}, examining an outside attack, provides a calculation for the expected number of blocks to be mined before the attacker has the longest chain. Here we want $\mathcal{L}$ to be the clock time, not the number of blocks but also, as will be shown, the precise level of $\mathcal{L}$ does not matter for the results below. Note also that $\mathcal{L}$  depends on $h_A$ and $H-\alpha h_A$, but we omit this consideration because this does not impact on the analysis below. } Let the expected costs of mounting an attack of length $\mathcal{L}$ by $A$ be $C_{attack}(\mathcal{L},A)$. (In what follows, as they play no explicit role, we drop the $\mathcal{L}$ and $A$ notation and refer simply to $C_{attack}$.) Then a majority attack is not worthwhile for $A$ if:
\begin{equation}\label{IC}
\Big(\frac{h^*_A}{\tau\,H}(R+\Phi) - c_A(h^*_A)\Big)\mathcal{L} \ge V_{attack} - C_{attack}\, .
\end{equation}
This is the generalized version of Budish's \textit{incentive compatibility} constraint.

\section{Costs of a Majority Attack}

We are now able to characterise the expected costs of mounting an attack, $C_{attack}$. If the attacker, $A$, supplies hash power of $h_A$ to the attacking branch, the costs of doing so are $c_A(h_A)$. As these are applied for the duration of the attack, the direct cost component of $C_{attack}$ is $c_A(h_A)\mathcal{L}$. 

This, however, is not a full description of $C_{attack}$ as $A$ will earn block rewards and transaction fees for blocks confirmed to the attack chain. Let $\tilde{R}$ and $\tilde{\Phi}$ be those expected benefits for each block mined during an attack. As $A$ is the only miner on the attacking branch, their expected number of blocks per unit of time is $\frac{h_A}{D}$. Thus, the second component of $C_{attack}$ is $-\frac{h_A}{D}(\tilde{R}+\tilde{\Phi})\mathcal{L}$. 

Putting these together, 
$$C_{attack}(h_A) = c_A(h_A)\mathcal{L} - \frac{h_A}{D}(\tilde{R}+\tilde{\Phi})\mathcal{L}\, .$$

\noindent However, during the attack, $A$ may choose a different level of hash power to minimize the cost of the attack. Let 
$\tilde{h}_A \equiv \arg\min_{h_A} \{ c_A({h}_A) - \frac{{h}_A}{D}(\tilde{R}+\tilde{\Phi}) \}$. To assure the success of the attack, however, $A$ needs to deploy at least $\underline{h}_A$. This implies that 
$$C_{attack} = \Big(c_A(\max\{\tilde{h}_A, \underline{h}_A\}) - \frac{\max\{\tilde{h}_A, \underline{h}_A\}}{D}(\tilde{R}+\tilde{\Phi})\Big)\mathcal{L}\, .$$

Then the incentive compatibility constraint~\eqref{IC} becomes
\begin{equation*}
\underbrace{\Big(\frac{h^*_A}{D}(R+\Phi) - c_A(h^*_A)\Big)\mathcal{L}}_{\text{opportunity cost of an attack}}\, \ge \, V_{attack} - \underbrace{\Big(c_A(\max\{\tilde{h}_A, \underline{h}_A\}) - \frac{\max\{\tilde{h}_A, \underline{h}_A\}}{D}(\tilde{R}+\tilde{\Phi})\Big)\mathcal{L}}_{\text{cost of mounting an attack}} \, .
\end{equation*}
On the left-hand side of this inequality are the non-attack returns that are forgone by $A$ should they attack. Re-arranging, this becomes:
\begin{equation}\label{IC2}
\underbrace{\Bigg(\frac{h^*_A}{D}(R+\Phi) - c_A(h^*_A) - \Big(\frac{\max\{\tilde{h}_A, \underline{h}_A\}}{D}(\tilde{R}+\tilde{\Phi})-c_A(\max\{\tilde{h}_A, \underline{h}_A\})\Big)\Bigg)\mathcal{L}}_{\text{net cost of an attack}} \ge V_{attack} \, .
\end{equation}
In their paper, Nakamoto (2008) stipulated that the cost of an attack is positive and increases with the length of the attack. Similarly, \cite{chiu2022economics} show that transactions with large enough value are vulnerable to double-spending.  \cite{budish2022economic} demonstrated that under some conditions, the cost of attack is zero, meaning that any transaction is vulnerable to an attack. 

For a negative net cost of attack, it is worth attacking even when $V_{attack}=0$, i.e., there is nothing to double spend. This may lead to attacks that are more difficult to detect.
Our generalized analysis shows that the key to signing the net cost of an attack is the difference between $A$'s expected profits if they do and do not attack, which we explore in further discussion.

\color{black}

\subsection{Linear mining costs}\label{sec:linear}

To build intuition, consider the case that \cite{budish2022economic} focuses on, where $c_i(h_i) = c\,h_i$ for any miner incumbent on the blockchain before the attack, $i \in \mathcal{M}'$, and for external hash power $c_i(h_i) = \kappa c\,h_i$ where $c > 0$ and $\kappa \ge 1$. 
As \cite{budish2022economic} notes, with linear and symmetric costs, the participation constraint holds for each miner with equality, giving rise to a free entry condition for all $i \in \mathcal{M}'$:
\begin{equation}\label{linear}
\frac{h_i}{\tau\,H}(R+\Phi) = c\,h_i \, . 
\end{equation}
Aggregating over $i \in \mathcal{M}'$, this yields
\begin{equation}\label{free.entry}
  R+ \Phi = c\,\tau\,H  \, .
\end{equation}

With linear costs, $\tilde{h}_A$ is not uniquely identified. Thus, we consider the cost of the attack for $\underline{h}_A$, the minimal required hash power for a successful attack of expected duration $\mathcal{L}$. The total direct costs of this attack are $c(\alpha + \kappa(1-\alpha))\underline{h}_A\mathcal{L}$. The indirect cost component is $\frac{\underline{h}_A}{\tau\,H}(\tilde{R}+\tilde{\Phi})\mathcal{L}$.
Due to the free entry condition~\eqref{linear}, the opportunity cost of attack, in this case, is zero, and the net cost of attack is equivalent to the cost of mounting the attack. 
Altogether, the incentive compatibility constraint~\eqref{IC2} becomes:
$$\underline{h}_A\Big(0 - \frac{1}{\tau\,H}(\tilde{R}+\tilde{\Phi}) + c(\alpha + \kappa(1-\alpha))\Big) \mathcal{L} \ge V_{attack}\, .$$
Further, using the aggregate free entry constraint~\eqref{free.entry}, this can be re-written as
$$\frac{\underline{h}_A}{\tau\,H}\Big( - (\tilde{R}+\tilde{\Phi}) + (R+\Phi)(\alpha + \kappa(1-\alpha))\Big) \mathcal{L} \ge V_{attack}\, ,$$
and further as 
\begin{equation}\label{linear.IC}
    \frac{\underline{h}_A}{\tau\,H}\Big( R+\Phi- (\tilde{R}+\tilde{\Phi}) + (R+\Phi)(\kappa-1)(1-\alpha)\Big) \mathcal{L} \ge V_{attack}\, .
\end{equation}
Note that the net cost of an attack on left-hand side of the inequality is decreasing in $\alpha$, and increasing in $\kappa$ and $R+\Phi- (\tilde{R}+\tilde{\Phi})$. Specifically, when $\tilde{R}+\tilde{\Phi}=R+\Phi$, the cost of an attack is zero when $\kappa=1$ or $\alpha=1$; that is when the cost of outside hashing power is the same as incumbent mining cost or if the attack is a fully inside attack. In such a case, the incentive compatibility constraint is violated and the blockchain is not secure as long as $V_{attack}>0$. It is secure, however, for $V_{attack}=0$. 

For other values of $\kappa$ and $\alpha$, the cost of attack is strictly positive. And that means that the incentive compatibility constraint holds for $V_{attack}=0$ and some positive values of $V_{attack}$. The blockchain is similarly secured for $V_{attack}=0$ and some positive values of $V_{attack}$, whenever $\tilde{R}+\tilde{\Phi}<R+\Phi$. This is because for such $\tilde{R}+\tilde{\Phi}$, the net cost of attack is positive for any values of $\kappa$ and $\alpha$, including $\kappa=1$ or $\alpha=1$. This is because such lower mining rewards on the attacking chain do not cover the cost of mining on that chain. Of course, if $V_{attack}$, it may justify bearing the positive net cost of the attack.

Interestingly, when $\tilde{R}+\tilde{\Phi}>R+\Phi$, then there exist values of  $\kappa$ and $\alpha$ such that the net cost of an attack is {\em negative}, e.g., $\kappa=1$ or $\alpha=1$ or values sufficiently close to it. In such a case, the incentive compatibility constraint is violated and the blockchain is not secure even when $V_{attack}=0$. That means that the attackers may rewrite the record on the blockchain not to double-spend, but to collect higher mining rewards.

This analysis provides some generalization over \cite{budish2022economic} who implicitly assumes that $\tilde{R}+\tilde{\Phi} = R+\Phi$ and $\alpha=0$. Note that, in such a case, 
the incentive compatibility constraint would become
 $\frac{h_A}{\tau\,H}(\kappa - 1)(R+\Phi)\mathcal{L} \ge V_{attack}\, .$
And as in \cite{budish2022economic}, the left-hand side (i.e., the cost of attack) will be zero if $\kappa = 1$. 

\subsection{General mining costs}

When there are non-linear and heterogeneous mining costs, $c_A(h_A)$ represents the total costs~$A$ has to expend per unit of time in order to apply $h_A$ in hash power to the network (whether attacking or not).  Similarly to the linear cost case, the incentive compatibility constraint~\eqref{IC2} is affected by the sign of $R+\Phi-(\tilde{R}+\tilde{\Phi})$. When $\tilde{R}+\tilde{\Phi}=R+\Phi$, the opportunity cost of attack is the same as $c_A(\tilde{h}_A)-\frac{\tilde{h}_A}{D}(\tilde{R}+\tilde{\Phi})\mathcal{L}$. Thus, the net cost of an attack is zero if the $\underline{h}_A$ constraint is not binding, i.e., when $\tilde{h}_A\geq \underline{h}_A$. If it is binding, then $A$ needs to deploy additional and more costly hashing power to assure the success of the attack. As a result, the cost of mounting the attack is larger than the opportunity cost and the net cost of the attack is strictly positive.\footnote{See Lemma~\ref{lem:equal.mining.rewards} in the Appendix for formal treatment of this result.}

When $\tilde{R}+\tilde{\Phi}<R+\Phi$, the net cost of the attack is always strictly positive. This is because even if $\underline{h}_A$ constraint is not binding, lower mining rewards on the attacking chain increase the cost of mounting of the attack, $C_{attack}$, above the opportunity cost.\footnote{See Lemma~\ref{lem:smaller.mining.rewards} in the Appendix for formal treatment of this result.} 

When the mining rewards are higher on the attacking blockchain, $\tilde{R}+\tilde{\Phi}>R+\Phi$ the cost of attack could be positive, zero or {\em negative}. Specifically, whenever the $\underline{h}_A$ constraint is not binding, the net cost of an attack is strictly negative. This is because higher mining rewards on the attacking chain decrease the cost of mounting an attack below the opportunity cost. The net cost of an attack may be strictly negative even in some cases when the $\underline{h}_A$ constraint is binding. A negative cost of attack means that the incentive compatibility constraint is violated --- and the security of the blockchain may be compromised --- even with $V_{attack}=0$.

If $\underline{h}_A\le h^*_A$, a \textit{purely internal attack} is feasible. Thus, the $\underline{h}_A$ constraint is not binding, and the cost of attack is negative.  Interestingly, even in this case,  outside hash power is applied to the attack, $\tilde{h}_A > h^*_A$. This is because higher mining rewards with a non-linear mining cost function lead to the use of more mining power, even though it is more expensive.

As $\underline{h}_A>\tilde{h}_A$ requires the deployment of hash power beyond cost minimizing point, the cost of mounting an attack increases. It is still lower than the opportunity cost until $\underline{h}_A$ reaches the break-even threshold, $\bar{\bar{h}}_A \equiv \{h_A\!>\!h^*_A\,|\,\frac{h^*_A}{\tau\,H}(R+\Phi) - c_A(h^*_A) = \frac{h_A}{\tau\,H}(\tilde{R}+\tilde{\Phi}) - c_A(h_A)\}$.  Beyond this point, $A$'s benefit of mounting the attack will be lower than their earnings had they not undertaken the attack. In this case, an attack involves a positive \textit{net} cost for $A$. For this reason, the incentive compatibility constraint can only be violated if $V_{attack} > 0$. Figure \ref{fig:attackprofits} provides a visual depiction of the hash power chosen by $A$ before and during an attack. 

This analysis allows us to establish Lemma~\ref{lem:conditional.attack}.

\begin{lemma}\label{lem:conditional.attack}
    When $\tilde{R}+\tilde{\Phi}>R+\Phi$, then there exists $0<\hat{h}<\frac{1}{2}H$ such that if $h^*_A\geq \hat{h}$, then the net cost of majority attack for $A$ is negative.
\end{lemma}

\begin{proof}
    First, note that $\underline{h}_A $ and $\bar{\bar{h}}_A$ are functions of the mining power~$A$ deploys on the incumbent blockchain, $h^*_A$. Moreover, for $h^*_A>0$ and a non-linear, non-decreasing $c_A(h_A)$, $\bar{\bar{h}}_A(h^*_A) > h^*_A$, and    given $\tilde{R}+\tilde{\Phi}>R+\Phi$. Further, $\bar{\bar{h}}_A(h^*_A)$ is increasing in $h^*$.

    Let $\hat{h}$ be characterized by $\hat{h}+\bar{\bar{h}}_A(\hat{h})=H$. Then $H>2\hat{h}$. Moreover, for any $h^*_A>\hat{h}$, $\bar{\bar{h}}_A(h^*_A)>\bar{\bar{h}}_A(\hat{h})=H-\hat{h}>H-h^*_A=\underline{h}_A(h^*_A)$. Therefore, the net cost of the majority attack for $A$ is negative.
\end{proof}\\

\noindent Lemma \ref{lem:conditional.attack} shows that if the mining rewards increase on the attacking chain, not only it is possible to conduct a majority attack at a negative net cost, but it is possible for a miner who does not have majority of the hash power in the incumbent blockchain.

The lemma hinges on the condition that the mining rewards on the attacking chain are higher than on the benchmark blockchain. Below we further establish that each of the cases $\tilde{R}+\tilde{\Phi}\lesseqqgtr R+\Phi$ can occur, depending on whether the attack is purely outside or partially inside attack, and congestion of the network.
\color{black}

\begin{figure}[t]
    \centering

\resizebox{\textwidth}{!}{

\tikzset{every picture/.style={line width=0.75pt}} 

\begin{tikzpicture}[x=0.75pt,y=0.75pt,yscale=-1,xscale=1]

\draw    (113,259) -- (616,259.5) ;
\draw [shift={(618,259.5)}, rotate = 180.06] [color={rgb, 255:red, 0; green, 0; blue, 0 }  ][line width=0.75]    (10.93,-3.29) .. controls (6.95,-1.4) and (3.31,-0.3) .. (0,0) .. controls (3.31,0.3) and (6.95,1.4) .. (10.93,3.29)   ; 
\draw    (113,259) -- (113,30.5) ;
\draw [shift={(113,28.5)}, rotate = 90] [color={rgb, 255:red, 0; green, 0; blue, 0 }  ][line width=0.75]    (10.93,-3.29) .. controls (6.95,-1.4) and (3.31,-0.3) .. (0,0) .. controls (3.31,0.3) and (6.95,1.4) .. (10.93,3.29)   ;
\draw  [dash pattern={on 0.84pt off 2.51pt}]  (300,126) -- (300,259.5) ;
\draw  [dash pattern={on 0.84pt off 2.51pt}]  (335,96) -- (335,259.5) ; 
\draw [line width=1.5]    (114,259.5) .. controls (225,202.5) and (263,75) .. (358,100) .. controls (453,125) and (553,255) .. (568,293) ;
\draw  [dash pattern={on 0.84pt off 2.51pt}]  (114,125) -- (407,125) ;
\draw  [dash pattern={on 0.84pt off 2.51pt}]  (411,125) -- (411,260) ;
\draw [line width=0.75]    (461,123) -- (438.51,142.68) ;
\draw [shift={(437,144)}, rotate = 318.81] [color={rgb, 255:red, 0; green, 0; blue, 0 }  ][line width=0.75]    (10.93,-3.29) .. controls (6.95,-1.4) and (3.31,-0.3) .. (0,0) .. controls (3.31,0.3) and (6.95,1.4) .. (10.93,3.29)   ;
\draw    (513,176) -- (463.57,214.77) ;
\draw [shift={(462,216)}, rotate = 321.89] [color={rgb, 255:red, 0; green, 0; blue, 0 }  ][line width=0.75]    (10.93,-3.29) .. controls (6.95,-1.4) and (3.31,-0.3) .. (0,0) .. controls (3.31,0.3) and (6.95,1.4) .. (10.93,3.29)   ;
\draw    (112,389) -- (410,389) ;
\draw [shift={(410,389)}, rotate = 180] [color={rgb, 255:red, 0; green, 0; blue, 0 }  ][line width=0.75]    (0,5.59) -- (0,-5.59)(10.93,-3.29) .. controls (6.95,-1.4) and (3.31,-0.3) .. (0,0) .. controls (3.31,0.3) and (6.95,1.4) .. (10.93,3.29)   ;
\draw [shift={(112,389)}, rotate = 0] [color={rgb, 255:red, 0; green, 0; blue, 0 }  ][line width=0.75]    (0,5.59) -- (0,-5.59)(10.93,-3.29) .. controls (6.95,-1.4) and (3.31,-0.3) .. (0,0) .. controls (3.31,0.3) and (6.95,1.4) .. (10.93,3.29)   ;
\draw    (410,389) -- (615.5,390) ;
\draw [shift={(615.5,390)}, rotate = 180.28] [color={rgb, 255:red, 0; green, 0; blue, 0 }  ][line width=0.75]    (0,5.59) -- (0,-5.59)(10.93,-3.29) .. controls (6.95,-1.4) and (3.31,-0.3) .. (0,0) .. controls (3.31,0.3) and (6.95,1.4) .. (10.93,3.29)   ;
\draw [shift={(410,389)}, rotate = 0.28] [color={rgb, 255:red, 0; green, 0; blue, 0 }  ][line width=0.75]    (0,5.59) -- (0,-5.59)(10.93,-3.29) .. controls (6.95,-1.4) and (3.31,-0.3) .. (0,0) .. controls (3.31,0.3) and (6.95,1.4) .. (10.93,3.29)   ;
\draw [line width=0.75]    (114,259.5) .. controls (236,213) and (235,127) .. (300,126) .. controls (365,125) and (520,261) .. (535,299) ;
\draw    (112,319) -- (295,319) ;
\draw [shift={(295,319)}, rotate = 180] [color={rgb, 255:red, 0; green, 0; blue, 0 }  ][line width=0.75]    (0,5.59) -- (0,-5.59)(10.93,-3.29) .. controls (6.95,-1.4) and (3.31,-0.3) .. (0,0) .. controls (3.31,0.3) and (6.95,1.4) .. (10.93,3.29)   ;
\draw [shift={(112,319)}, rotate = 0] [color={rgb, 255:red, 0; green, 0; blue, 0 }  ][line width=0.75]    (0,5.59) -- (0,-5.59)(10.93,-3.29) .. controls (6.95,-1.4) and (3.31,-0.3) .. (0,0) .. controls (3.31,0.3) and (6.95,1.4) .. (10.93,3.29)   ;
\draw  [draw opacity=0][fill={rgb, 255:red, 255; green, 255; blue, 255 }  ,fill opacity=1 ] (153,299) -- (257,299) -- (257,339) -- (153,339) -- cycle ;
\draw  [draw opacity=0][fill={rgb, 255:red, 255; green, 255; blue, 255 }  ,fill opacity=1 ] (300,292) -- (339,292) -- (339,345) -- (300,345) -- cycle ;
\draw    (295,319) -- (613.5,317) ;
\draw [shift={(613.5,317)}, rotate = 179.64] [color={rgb, 255:red, 0; green, 0; blue, 0 }  ][line width=0.75]    (0,5.59) -- (0,-5.59)(10.93,-3.29) .. controls (6.95,-1.4) and (3.31,-0.3) .. (0,0) .. controls (3.31,0.3) and (6.95,1.4) .. (10.93,3.29)   ;
\draw [shift={(295,319)}, rotate = 359.64] [color={rgb, 255:red, 0; green, 0; blue, 0 }  ][line width=0.75]    (0,5.59) -- (0,-5.59)(10.93,-3.29) .. controls (6.95,-1.4) and (3.31,-0.3) .. (0,0) .. controls (3.31,0.3) and (6.95,1.4) .. (10.93,3.29)   ;
\draw  [draw opacity=0][fill={rgb, 255:red, 255; green, 255; blue, 255 }  ,fill opacity=1 ] (346,301) -- (504.5,301) -- (504.5,341) -- (346,341) -- cycle ;

\draw (608,263.9) node [anchor=north west][inner sep=0.75pt]    {$h_{A}$};
\draw (292,263.9) node [anchor=north west][inner sep=0.75pt]    {$h_{A}^{*}$};
\draw (328.5,260.65) node [anchor=north west][inner sep=0.75pt]    {$\tilde{h}_{A}$};
\draw (61,307.4) node [anchor=north west][inner sep=0.75pt]    {$\underline{h}_{A}$};
\draw (155,302) node [anchor=north west][inner sep=0.75pt]   [align=left] {\begin{minipage}[lt]{67.93pt}\setlength\topsep{0pt}
\begin{center}
pure inside\\attack feasible
\end{center}

\end{minipage}};
\draw (512,157.4) node [anchor=north west][inner sep=0.75pt]    {$R+\Phi $};
\draw (461,98.4) node [anchor=north west][inner sep=0.75pt]    {$\tilde{R}+\tilde{\Phi }$};
\draw (184,393) node [anchor=north west][inner sep=0.75pt]   [align=left] {Attack even if $\displaystyle V_{attack} =0$};
\draw (423,393) node [anchor=north west][inner sep=0.75pt]   [align=left] {Attack requires $\displaystyle V_{attack}  >0$};
\draw (47,38) node [anchor=north west][inner sep=0.75pt]   [align=left] {\begin{minipage}[lt]{40.14pt}\setlength\topsep{0pt}
\begin{center}
Attacker\\Flow\\Profits
\end{center}

\end{minipage}};
\draw (406,260.4) node [anchor=north west][inner sep=0.75pt]    {$\bar{\bar{h}}_{A}$};
\draw (368,301) node [anchor=north west][inner sep=0.75pt]   [align=left] {\begin{minipage}[lt]{83.23pt}\setlength\topsep{0pt}
\begin{center}
outside resources\\applied in attack
\end{center}
\end{minipage}};

\end{tikzpicture}
}
\caption{\textbf{Attacker Profits and Costs when $\tilde{R}+\tilde{\Phi}>R+\Phi$}}
\label{fig:attackprofits}  
\end{figure}
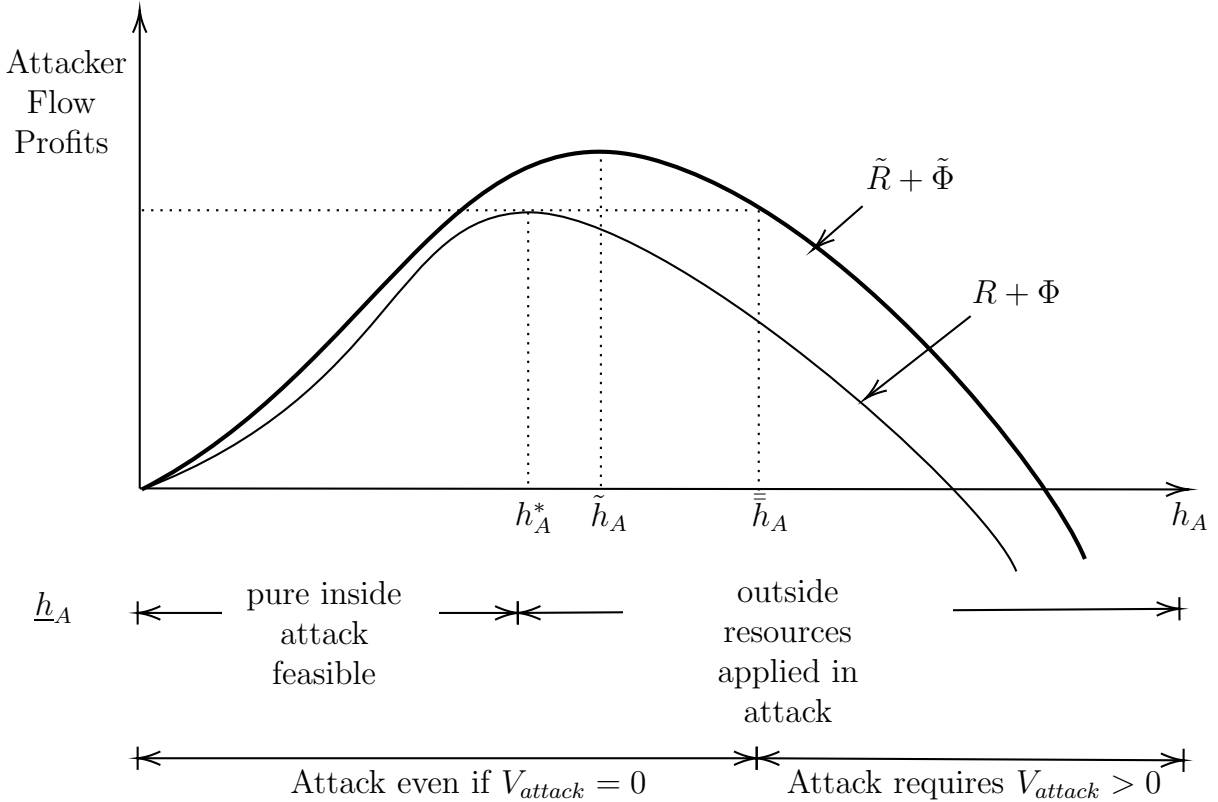

\subsection{Transaction fees collected during an attack}

The above analysis shows that a key driver of the costs of a majority attack, and thus security of the blockchain, are the mining rewards, i.e., block rewards and transaction fees, during an attack (i.e., $\tilde{R}$ and $\tilde{\Phi}$). Note that, with Bitcoin, the block reward rarely changes and is fixed by the protocol. As mentioned earlier, our focus here is on internal mechanisms that mediate the cost of attacks, and therefore in our analysis we assume that the exchange rate, as an external societally-mediated factor, is not affected by the attack. In result, $\tilde{R}$ does not differ from $R$. For this reason, we focus here instead on what happens to transaction fees.

Every transaction, $t$, is characterized by amount of space it takes, $\sigma_t>0$ and transaction fee {\em per unit of space} offered by a user, $\varphi_t\geq 0$. Transactions arrive randomly and independently, and on average, the sum of $\sigma_t$ for all transactions arriving in a unit of time is $\sigma$. 

Blocks have a fixed capacity of $b$ that is set by the protocol. For Bitcoin, the block size is limited to 1MB, which can house over 2,000 transactions. If blocks, on average, are not at capacity, i.e., $\tau\, \sigma \le b$, then users gain no advantage in terms of the speed of transaction processing and so will not find it optimal to offer a fee above zero. In this case, $\Phi = 0$. Thus, if the blocks are not filled both during the attack and without an attack, then $\tilde{\Phi}=\Phi=0$, and then $\tilde{R}+\tilde{\Phi}=R+\Phi$.

On the other hand, if the demand for blockchain transactions is higher than the supply, i.e., $\tau\, \sigma>b$, then only transactions paying the highest fee per unit of space get included in the block. In this case, to determine $\Phi$, order all transactions available at a point in time (e.g., transactions that arrived during the last $\tau$ units of time since the last block) by $\varphi$ in such a way that $\varphi_1\geq \varphi_2 \geq ... $ with some of these inequalities strict. To avoid the knapsack problem, it is assumed that all transactions have the same size of 1, i.e., $\sigma_t=1$ for all $t$. Then, $\sigma$ is simply the average number of transactions per unit of time, and a block has the capacity to hold $b$ transactions. Then the fees collected with this block are $\Phi=\sum_{t=1}^{b}\varphi_t$. This is equivalent to picking the $b$ highest fees from  $\sigma$ transactions arriving.

Given this, we can prove the following.
\begin{lemma}\label{lem:higher.fees}
    Suppose that users do not alter their transaction fee submissions during an attack and that $h_A \in (\frac{1}{2}H,H)$. If $\tau\, \sigma > b$, $\tilde{\Phi} > \Phi$.
\end{lemma}
\begin{proof}
    Note that unless $h_A \ge H$, the total number of blocks confirmed per period on the attacking chain is strictly less than the number of blocks confirmed per period when there is no attack. As users do not alter their transaction fee submissions, the attacking chain has a strictly lower total block capacity per period. Thus, there exist transactions that would have been confirmed without the attack but are not confirmed on the attacking chain. As the attacker chooses transactions for each block to maximize total fees, this implies that total fees on the attacking chain per block exceed those that would otherwise arise per block on the benchmark blockchain without the attack. 
    
    More formally, the expected time for a block to be added on the attacking blockchain is $\tau\frac{H}{h_A}>\tau$. Every block still has a capacity of $b$, but more transactions arrive between block confirmations, $\tau\frac{H}{h_A} \sigma>\tau \sigma$. So now, the attacker is choosing the $b$ highest $\varphi$'s from a larger set. Their sum is expected to be larger with the larger set, by order statistics. (Unless all transactions have the same fee per space, which may happen when the fee is 0.) 
    
    And specifically, if the support of $\varphi$'s is larger than a single point, there is a strictly positive probability that among the additional $(\frac{H}{h_A}-1)\tau \sigma$ transactions arriving because of the delay, there is a transaction with $\tilde{\varphi}>\varphi_{b}$. That would increase fees collected under attack, $\sum_{t=1}^{b-1} \varphi_t + \tilde{\varphi}>\sum_{t=1}^b \varphi_t$. Therefore, the expected average fees on the attacking chain are larger than without an attack, $\tilde{\Phi}>\Phi$.
\end{proof}
\\

\noindent The simple intuition behind this result is that, under the attack, fewer blocks are added to the blockchain per unit of time, and therefore the capacity of the blockchain (i.e., supply of processed transactions) is lower, while the demand for processing transactions remains the same. When mining a block under attack, $A$ has the same transactions available as they would have without an attack and some new arrivals, which may offer higher fees. (These ``new arrivals" could have been picked up by other miners without the attack.)

It is useful to note that, for the same reason, fees on the honest branch during the attack also increase. However, miners do not collect these fees as the honest branch becomes orphaned when the attacking branch eventually becomes the longest chain.

Given this, we can demonstrate the following:

\begin{proposition}
There exists $0<\hat{h}<\frac{1}{2}H$, such that if for some miner $i \in \mathcal{M}'$,  $h^*_i> \hat{h}$, then an equilibrium without a majority attack does not exist even if $V_{attack} = 0$.
\end{proposition}

\begin{proof}
For any $h^*_i>0, \underline{h}_A(h^*_i)<H$. And therefore, $i$ can conduct a successful attack with $h_A<H$. Then by Lemma~\ref{lem:higher.fees}, $\tilde{R}+\tilde{\Phi}>R+\Phi$ on the attacking chain. 

Now, if also $h^*_i>\hat{h}$, where $\hat{h}\equiv \{h\,|\,h+\bar{\bar{h}}(h)=H\}$, then by Lemma~\ref{lem:conditional.attack}, $\hat{h}<\frac{1}{2}H$ and the net cost of attack for~$i$ is negative. Thus, even when $V_{attack}=0$, the incentive compatibility constraint for $i$ is violated, and there is no equilibrium without a majority attack.
\end{proof}
\\

\noindent If for every active miner $i \in \mathcal{M}'$, $h^*_i<\hat{h}$, then the incentive compatibility constraint holds for all active miners. In this case, a fully or partially inside majority attack could only proceed if $V_{attack} > 0$. 
However, in Section~\ref{sec:renting}, we show that if it is possible to rent inside hash power, it is not necessary for a miner to have more than $\hat{h}$ of incumbent power -- any miner can successfully attack at a negative cost.

\subsection{Costs of an outside attack}

The above analysis demonstrates that the scenario of a purely outside attack is somewhat of a special case in understanding the security of the blockchain. For that reason, it is useful to consider this case more carefully. In so doing, we focus on the linear cost case.

Consider a purely outside attack where $h^*_A=0$ and $h_A>H$. Since more blocks will be produced per unit of time, and the transactions arrive at the same rate, by the same logic as in Lemma~\ref{lem:higher.fees}, for fully outside attacks $\tilde{\Phi}<\Phi$. Thus, the mining rewards are strictly lower on the attacking chain, and by earlier analysis, the net cost of attack is strictly positive.

Let $m=\frac{h_A}{H}>1$. Then, the expected direct cost of mining per unit of time is $c h_A=c\, m\,  H$, where, by the free entry condition, $c=\frac{R+\Phi}{\tau H}$. So the expected cost of mining during the attack is $c h_A=\frac{m}{\tau}(R+\Phi)$. The mitigating mining rewards during the attack bring, in expectation, $\frac{m}{\tau}(R + \tilde{\Phi})$. So the net cost of the attack per unit of time is $\frac{m}{\tau}(\Phi-\tilde{\Phi})$. The total expected cost of the attack is: 
$$\frac{\mathcal{L}(m)}{\tau} m (\Phi-\tilde{\Phi}) > 0\, .$$
The left-hand side is positive since for $h_A>H, \tilde{\Phi}>\Phi$.

It would seem that with an attack from outside, the attacker's cost increases with the attack's length. However, the attacker can always limit the number of blocks mined in the attack to one more than the number of blocks mined on the honest blockchain. That is, the moment that $A$ mines one more block than the honest blockchain, the attack ends.\footnote{Recall there is no escrow period as that is a convention outside of Proof of Work protocols. If there were an escrow rule, this would not change the logic. Suppose an attacker is forced to wait until the honest blockchain produces at least $w$ blocks before revealing their longer chain. If $A$ mined continuously during that time, they would mine in expectation $m\, w$ blocks, possibly more than $w+1$. However, the attacker does not need to mine more than $w+1$. Thus, they can shut down mining after reaching $w+1$, and wait until the honest block reaches $w$ to reveal their longer chain. This will allow them to save on the mining cost that is not necessary for the success of the attack.} Thus, $A$ produces only one more block than the honest blockchain no matter how long the attack lasts. Therefore, more properly, the expected number of blocks $A$ mines in the attack is $\frac{\mathcal{L}}{\tau}+1$. The expected mining cost (with shutting down if necessary) is $\left(\frac{\mathcal{L}}{\tau}+1\right)(R +\Phi)$. The mitigating mining rewards are $\left(\frac{\mathcal{L}}{\tau}+1\right)(R +\tilde{\Phi})$. 

Note that with $\left(\frac{\mathcal{L}}{\tau}+1\right) \tilde{\Phi}$, $A$ processes all the transactions that would go into the honest blocks during $\frac{\mathcal{L}}{\tau}$, and they also create one additional block of ``second tier" transactions paying strictly less. That is, the total fees collected by the attacker throughout the attack are $\frac{\mathcal{L}}{\tau} \Phi+\tilde{\Phi}(\mathcal{L})$. Note also that $\tilde{\Phi}(\mathcal{L})$ is increasing in the length of the attack $\mathcal{L}$. This is because there are more second-tier transactions, and the highest paying $b$ transactions of second-tier transactions pay more. Therefore, the total expected cost of the majority attack from the outside is:
$$
\Phi-\tilde{\Phi}(\mathcal{L})>0\, .
$$
The net cost of the attack is still positive, but it decreases the longer the attack.

\subsection{Users adjusting their transaction fee bids}

Until now, we have assumed that the transactions (including fees offered) arrive in the same way during the attack as they would without an attack. But, in fact, the users may adjust the fees they offer based on the congestion they observe.

We assume that the attacker keeps their branch of the fork secret until ready to reveal the longest chain. Only the honest blockchain (first the benchmark blockchain and then the honest non-attacking branch) is visible to the users before the attack is executed, and users respond only to this. 

In the case of a fully outside attack, users see no change and do not adjust their bids. In the case of an inside attack, the time between the blocks on the honest chain increases, decreasing the supply of transactions recorded per unit of time. Therefore, users who adjust their bids, will adjust upward, increasing the difference, $\tilde{\Phi}>\Phi$. 

\subsection{Comparison with selfish mining}

It is useful to compare this result with the work by \cite{eyal2018majority}. That paper demonstrates that a non-majority miner controlling between one-third and one-half of total hash power can find it optimal to deviate from the honest strategy and mine a private chain. The difference here is that in the mechanism for a majority attack, the attacker controls more hash power than all other miners during the attack. 
In this regard, the two results are complementary, and they are driven by different forces. 

An incumbent miner engaging in selfish mining creates a secret chain in the hope of mining more blocks in a given time interval than the benchmark and thus collects additional block rewards, weighing it against the risk of being outrun by the remaining majority of miners. If the same miner engages in a partially inside majority attack instead, he mines fewer blocks than the benchmark, incurring additional mining costs of deploying outside hashing power but collecting higher transaction fees.
\color{black}

\section{Model Extensions}

\subsection{Renting inside hash power}\label{sec:renting}
As argued by \cite{bonneau2016buy}, it may be possible for the attacker to rent hash power from miners currently participating in the blockchain. To rent $h_{rent}$ of hash power currently in use, the attacker needs to compensate the owner, i.e., pay
$$
r(h_{rent}) \gtrsim\frac{h_{rent}}{D}(R+\Phi)-c_{rent}(h_{rent})\, .
$$ 
Then the cost of the attack per unit of time becomes
$$
c_A(h_A)-\frac{h_A}{D}(R+\tilde{\Phi})+c_{rent}(h_{rent})- \frac{h_{rent}}{D}(R+\tilde{\Phi}) +r(h_{rent})$$
$$=c_A(h_A)-\frac{h_A}{D}(R+\tilde{\Phi})+\frac{h_{rent}}{D}(\Phi-\tilde{\Phi})
$$
where $h_A$ is the power the attacker directly provides himself, and $h_{rent}$ is the power rented from the current miners.

Notice that renting hashing power that is currently in use affects $\underline{h}_A$. The minimal hashing power needed for a majority attack when renting internally is $\underline{h}^A(h_{rent})\gtrsim H-h^*_A-h_{rent}$. When~$i$ rents $h_{rent}>\frac{1}{2}H-h^*_A$, then $h^*_A+h_{rent}>\underline{h}_A(h_{rent})$. This means that the attacker can conduct a successful majority attack using no more than their original $h^*_A$.

With $\tilde{\Phi}>\Phi$, the possibility of renting internal hashing power makes the attacks less costly. 

\begin{proposition} With the possibility of renting internal hashing power, any miner can successfully attack at a negative cost of attack, i.e., the attack is beneficial even if $V_{attack}=0$.
\end{proposition}

\begin{proof}
Acknowledging that $\tilde{R}=R$, when $\tilde{\Phi}>\Phi$, the per-unit-of-time cost of attack with renting $h_{rent}$ of hashing power already used in the blockchain is 
$$c_A(h_A)-\frac{h_A}{D}(R+\tilde{\Phi})+\underbrace{\frac{h_{rent}}{D}(\Phi-\tilde{\Phi})}_{<0}
$$
Since $A$ can rent $h_{rent}>\frac{1}{2}H-h^*_A$, then successfully attacking with $h^*_A$ of own hash power is an option. Optimizing own and rented hashing power to successfully attack at a minimal cost yields:
\begin{multline*}
C_{attack}\leq \\ \left(c_A(h^*_A)-\frac{h^*_A}{D}(R+\tilde{\Phi})+\frac{h_{rent}}{D}(\Phi-\tilde{\Phi})\right)\mathcal{L}<\left(c_A(h^*_A)-\frac{h^*_A}{D}(R+\tilde{\Phi})\right)\mathcal{L}\\<\left(c_A(h^*_A)-\frac{h^*_A}{D}(R+{\Phi})\right)\mathcal{L}
\end{multline*}
which means that the incentive compatibility constraint is violated, i.e., it is beneficial to attack even if $V_{attack}=0$.
\end{proof}

\subsection{Adjusting mining difficulty during an attack}\label{sec:difficulty}

The previous analysis held the blockchain's difficulty, $D$, fixed at $D = \tau H$ where $H$ is the previous aggregate hash power and $\tau$ is the targeted average time between block confirmations. However, as already noted, during an attack, when $\underline{h}_A < H$, the time between block confirmations (on both the honest and attacking chain) increases. If difficulty were to adjust during the attack, it would restore the timing of blocks back to $\tau$. Because the aggregate hash power on the honest and attack chain is different, the two chains would have different difficulties by the time the attack chain is made public. Here we explore how this changes the operation and incentives of an attack.

A simple conjecture might be that if difficulty adjusts immediately, something we can theorize about even though it cannot occur in practice, the two competing chains' difficulties would adjust and either one could be the longest chain and the attack chain, despite having more hash power, would have no advantage. However, the longest chain rule is not, in fact, purely a convention based on the longest chain. Instead, when two competing fork branches have different difficulties, the generalized longest chain rule becomes the {\em heaviest chain} rule. That is, the number of blocks is weighted by difficulty so that the branch with the most computational power behind it is chosen.\footnote {The difficulty of each fork is observable to all miners as it is captured in the level of the computational problem each faces. See e.g., https://learnmeabitcoin.com/technical/longest-chain}

Given that, we explore what happens when difficulty adjusts during an attack. In Bitcoin, difficulty adjusts every 2600 blocks. To adjust, the algorithm takes the average (reported) clock time that it has taken to mine these 2600 blocks, and if that number is different than 10 min, it adjusts the difficulty to such a number that would yield 10 min given the power used over the past 2600 blocks. Note that this adjustment is purely retrospective and does not account for the difference between more recent and old changes. Therefore, if twice as much hashing power was present over the last ten blocks, the algorithm won't assume this is the hashing power going forward. It will take the average over the 2600 blocks.

Suppose that the attack started $d$ blocks before the change of difficulty. On average, it took $\tau (2600-d)$ time to mine the blocks on the pre-attack blockchain. Then, the average time to mine a block on the attacking branch with power $h_A$ and difficulty $D=H\tau$ is $Y_A=\frac{1}{\theta_A}=\frac{D}{h_A}$. So at the time of the difficulty adjustment, the total average time it took to mine 2600 blocks on this branch is $\tau (2600-d)+d\tau \frac{H}{h_A}=\tau \big((2600-d)+d\frac{H}{h_A}\big)$. 

With that, the new difficulty is 
\begin{equation*}
D'_A=D\frac{2600}{(2600-d)+d\frac{H}{h_A}}\,.
\end{equation*}
Note that $D'_A<D$ when $h_A<H$ and $D'_A>D$ when $h_A>H$. Moreover,the new expected time between blocks is $$
Y'_A=\frac{D'_A}{h_A}=\tau \frac{H}{h_A}\frac{2600}{(2600-d)+d\frac{H}{h_A}}\, .
$$
When $d=2600$, then $Y'_A=\tau$. But when $d<2600$, then
$Y'_A>\tau$  when $h_A<H$ and $Y'_A<\tau$  when $h_A>H$.

Whenever $h_A\geq \underline{h}_A$, $A$'s attack eventually succeeds with probability 1. This is because, by the definition of $\underline{h}_A$, $h_A \geq \underline{h}_A$ will make the attacking branch heavier, even if not longer. Therefore, we can express the clock-counted expected length of the attack as $\mathcal{L}=\mathcal{L}_{preD}+\mathcal{L}_{postD}$, where $\mathcal{L}_{preD}$ is the expected length of the attack before the difficulty change ($\mathcal{L}_{preD}=d \tau$) and $\mathcal{L}_{postD}$ is the expected length of the attack after the difficulty change.

In such a case, the cost of the attack is
\begin{equation*}
C_{attack}(h_A)=\mathcal{L}_{preD}[c_A(h_A)-\frac{h_A}{D}(R+\tilde{\Phi})]+\mathcal{L}_{postD}[c_A(h_A)-\frac{h_A}{D'_A}(R+\tilde{\Phi}_{postD})]
\end{equation*}
where $D'_A$ is the difficulty on the attacking branch after adjustment, and $\tilde{\Phi}_{postD}$ are the expected fees per block on the attacking branch after the adjustment.

So, if $h^A_i<H$, then $D'_A<D$ and the attacking branch mines more blocks than before the adjustment change, $\frac{h_A}{D'_A}>\frac{h_A}{D}$. Thus, the attacker gets more block rewards per unit of time after the difficulty adjustment. But the expected fees are lower per block after the adjustment than before, $\tilde{\Phi}_{postD}<\tilde{\Phi}$. Nonetheless, they are still (weakly) higher than on the benchmark blockchain, $\tilde{\Phi}_{postD}\geq \Phi$, because $Y_A'\geq \tau$, i.e., the blocks are created more slowly. Thus, the capacity of the blockchain on the attacking chain is weakly smaller than on the benchmark blockchain. It is strictly smaller, and $Y_A'>\tau$ when $d<2600$. Altogether, $\frac{h_A}{D'_A}(R+\tilde{\Phi}_{postD})>\frac{h_A}{D_A}(R+\tilde{\Phi}_{postD})>\frac{h_A}{D_A}(R+\Phi)$. Therefore, the cost of attacking for $h^*_A>0$ is lower than the opportunity cost even with difficulty adjustment: 
$$\min_{h_A} C_{attack}(h_A<H)\leq C_{attack}(h^*_A)<\left(c_A(h^*_A)-\frac{h^*_A}{D}(R+\Phi)\right)\mathcal{L}\, ,$$
confirming our baseline result.

If $h_A>H$, i.e., the attack is purely outside, then $D'_A>D$, which means that the number of blocks the attacker mines after the difficulty change is lower than before the difficulty change, $\frac{h_A}{D'_A}<\frac{h_A}{D}$. Nonetheless, the new expected time between blocks is still lower than on the benchmark blockchain, $Y'_A<\tau$, which means $\tilde{\Phi}_{postD}<\Phi$. With that, 
$$C_{attack}(h_A>H)>\left(c_i(h_A)-\frac{h_A}{D}(R+\Phi)\right)\mathcal{L}\,.$$
The cost of attack is positive when $c_A(h_A)\geq \frac{h_A}{D}(R+\Phi)$ which holds by the free entry condition that implies that the attacker is weakly less efficient than the marginal miner.

\subsection{Miners adjusting their participation}

For our earlier results, we assumed that an attack was short-run and, thus, there was no entry or exit of honest miners. Accounting for changes in the participation constraint of honest miners impacts the results in an important way.

Many potential miners may be ready to join if the mining becomes marginally more profitable. At the outset, prior to an attack, the blockchain profit of a marginal miner is zero, and, thus, additional miners do not enter. In case of an outside attack, the observable honest branch does not differ from the benchmark blockchain, and hence there is no additional entry.

In the case of a non-purely outside attack, the honest branch of the fork slows down, increasing mining rewards as fees per block increase (whether adjusted or not). That will encourage new mining power to enter. 

If the potential entrant miners on the fringe have at least the same mining efficiency as the marginal miner in the benchmark blockchain, then the honest branch of the blockchain will grow up to $H$ in mining power. That changes the requirement on $\underline{h}_A$, and basically requires that $\underline{h}_A \gtrsim H$, making it a fully outside attack. 

If the entry of these new, equally efficient miners happens immediately, the cost of the attack is larger than the opportunity cost. This is because the cost of the attack is positive (because it's effectively a fully outside attack), and the opportunity cost is weakly negative (the miner had a non-negative payoff from mining on the original blockchain).

If the entry of these new, equally efficient miners is delayed, the cost of the attack is negative until the honest chain reaches $H$. This is because $A$ can conduct a majority attack with $h_A<H$. Once the honest chain reaches $H$, the attacker needs to deploy $h_A>H$, and the cost of the attack becomes positive. Whether the total cost of the attack is positive or negative depends on the length of the attack. Since the speed at which the new miners enter the honest chain is independent of $h_A$, it may be optimal to increase $h_A$ to increase the chance that the attack finishes before the honest branch reaches $H$ and the cost of the attack becomes positive.

If the newly entering miners on the honest chain are less efficient than those on the benchmark blockchain, the honest chain will not reach $H$ before the marginal miner breaks even. With less than $H$ on the honest branch, $A$ can conduct a successful majority attack with $h_A<H$, and thus with the cost of an attack less than the opportunity cost.

This highlights that attacks that take place beyond the short-run will involve additional entry effects that may raise longer-run attack costs beyond those expressed on our baseline result. However, while the discussion above illustrates some of those potential nuances, a full equilibrium analysis of the long-run model, complete with expectations regarding the success of the attack and post-attack behavior, is required. This extension is left for future work.

\section{Proof of Stake with Nakamoto Consensus}

We now turn to consider a version of Proof of Stake consensus (specifically, permissionless Proof of Stake blockchains with Nakamoto Consensus) and demonstrate that the same outcomes as those derived for a majority attack on Proof of Work blockchains also apply for those protocols.\footnote{There are two broad variants of Proof of Stake. The first one to emerge was based on Nakamoto consensus whereby, if there were forks, the chain that validators extended would be the one with the most blocks. This was the consensus mechanism of PeerCoin. This is not the most widely used variant today. That is based on Byzantine Fault Tolerance (BFT) and is not vulnerable to the majority attacks as analyzed in this paper. For more on the distinction between these approaches see \cite{gans2023consensus}. \cite{amoussou2019rationals}, \cite{auer2021ledgers} and \cite{halaburda2021economic} discuss attack vulnerabilities of BFT blockchains.}

\subsection{Equilibrium without an Attack}

Let $\mathcal{V}$ be the set of validators with individual validator $i \in \mathcal{V}$. The stake, in blockchain-native coins) of $i$ is $s_i$. Staking involves a cost associated with locking up capital in the network. Let $r(s_i)=r\, s_i e_s$ -- average cost of staking (locking up the capital) per unit of clock time, where $e_s$ is the (expected) exchange rate and $r$ could be interpreted as an interest rate.

The staking process and block proposer selection proceeds as follows:
\begin{itemize}
\item the (approximate) time between the blocks is set by the protocol to be $\tau_s$
\item the protocol calls validator $i$ to mint a block with probability proportional to their stake; each draw independent
\begin{itemize}
\item every $\tau_s$,  validator $i$ mints a block with probability $\frac{s_i}{S}$ 
\item within a time $K=k\, \tau_s$ validator~$i$ expects to mint $k\frac{s_i}{S}$ block, since each block's draw independent
\end{itemize}
\item if a validator is called but does not propose a block, that is, they are a no-show, the block missing, next node selected according to staking proportion for next $\tau_s$
\item if a validator is called and proposes a valid block they potentially receive newly minted tokens and transaction fees with expected values in fiat currency of $R_s$ and $\Phi_s$ respectively. (Transaction fees are determined in the same way as under Proof of Work.)
\end{itemize}
Thus, each validator can choose a stake level that gives them a higher probability of being the stake proposer. That is, if the total amount staked is $S=\sum_{i\in\mathcal{V}} s_i$, then the probability of being selected to propose a stake in any epoch is $\frac{s_i}{S}$. 

Given this, for a clock interval $K=k\, \tau_s$, each validator $i$ chooses $s_i$ to maximize expected profits:
$$k\frac{s_i}{S}(R_s+\Phi_s)-k\,\tau_s\, r\, s_i e_s$$
If $s^*_i$ is the maximized stake amount, then the participation constraint for each $i$ is:
$$\frac{1}{S}(R_s+\Phi_s) \ge \tau_s\,r\, e_s$$
Let $\mathcal{V}' \subseteq \mathcal{V}$ be the set of validators who satisfy this participation constraint. Given the linear costs that arise under Proof of Stake, free entry implies that the participation constraint will bind for all $i \in \mathcal{V}'$ so that:
$$\frac{1}{S}(R_s+\Phi_s) = \tau_s\,r\, e_s \implies R_s+\Phi_s = \tau_s\,r\, e_s\,S$$

\subsection{Attack by a Majority Staker}

As validators must be accepted into the protocol, the relevant mode of attack is an inside attack.\footnote{Previous analyses of Proof of Stake attacks assumed an outside attack; e.g., \cite{gans2021consensus}, \cite{halaburda2022microeconomics}.} We now describe that attack as depicted in Figure \ref{fig:PoSattack}.

\begin{figure}[t]
\centering
\resizebox{\textwidth}{!}{

\tikzset{every picture/.style={line width=0.75pt}} 

\begin{tikzpicture}[x=0.75pt,y=0.75pt,yscale=-1,xscale=1]

\draw  [fill={rgb, 255:red, 155; green, 155; blue, 155 }  ,fill opacity=0.6 ] (490.25,74) -- (534.75,74) -- (534.75,112) -- (490.25,112) -- cycle ;
\draw  [fill={rgb, 255:red, 155; green, 155; blue, 155 }  ,fill opacity=0.6 ] (61,73) -- (105.5,73) -- (105.5,111) -- (61,111) -- cycle ;
\draw    (296.5,191) -- (675,190.5) ;
\draw [shift={(678,190.5)}, rotate = 179.92] [fill={rgb, 255:red, 0; green, 0; blue, 0 }  ][line width=0.08]  [draw opacity=0] (8.93,-4.29) -- (0,0) -- (8.93,4.29) -- cycle    ;
\draw    (106.5,92.95) -- (124,92.95) ;
\draw    (122.5,191) -- (212.5,191) ;
\draw    (207.5,191) -- (296.5,191) ;
\draw [line width=1.5]  [dash pattern={on 1.69pt off 2.76pt}]  (83.5,150) -- (118,150) ;
\draw [shift={(122,150)}, rotate = 180] [fill={rgb, 255:red, 0; green, 0; blue, 0 }  ][line width=0.08]  [draw opacity=0] (11.61,-5.58) -- (0,0) -- (11.61,5.58) -- cycle    ;
\draw [line width=1.5]  [dash pattern={on 1.69pt off 2.76pt}]  (83.5,112.5) -- (83.5,150) ;
\draw    (537.5,23.5) .. controls (507.28,43.98) and (516.03,55.89) .. (536.87,67.6) ;
\draw [shift={(538.5,68.5)}, rotate = 208.61] [color={rgb, 255:red, 0; green, 0; blue, 0 }  ][line width=0.75]    (10.93,-3.29) .. controls (6.95,-1.4) and (3.31,-0.3) .. (0,0) .. controls (3.31,0.3) and (6.95,1.4) .. (10.93,3.29)   ;
\draw    (622,174.5) .. controls (598.84,192.84) and (577.08,191.61) .. (567.04,177.13) ;
\draw [shift={(566,175.5)}, rotate = 59.3] [color={rgb, 255:red, 0; green, 0; blue, 0 }  ][line width=0.75]    (10.93,-3.29) .. controls (6.95,-1.4) and (3.31,-0.3) .. (0,0) .. controls (3.31,0.3) and (6.95,1.4) .. (10.93,3.29)   ;
\draw  [fill={rgb, 255:red, 155; green, 155; blue, 155 }  ,fill opacity=0.6 ] (2,73) -- (46.5,73) -- (46.5,111) -- (2,111) -- cycle ;

\draw  [dash pattern={on 0.84pt off 2.51pt}]  (47.5,92.5) -- (60,92.5) ;
\draw  [fill={rgb, 255:red, 155; green, 155; blue, 155 }  ,fill opacity=0.6 ] (183,73) -- (227.5,73) -- (227.5,111) -- (183,111) -- cycle ;
\draw    (228.5,92.95) -- (246,92.95) ;
\draw  [fill={rgb, 255:red, 155; green, 155; blue, 155 }  ,fill opacity=0.6 ] (124,73) -- (168.5,73) -- (168.5,111) -- (124,111) -- cycle ;
\draw    (169.5,92.5) -- (182,92.5) ;
\draw    (107.5,149.95) -- (125,149.95) ;
\draw  [fill={rgb, 255:red, 155; green, 155; blue, 155 }  ,fill opacity=0 ] (184,130) -- (228.5,130) -- (228.5,168) -- (184,168) -- cycle ;
\draw    (229.5,149.95) -- (247,149.95) ;
\draw  [fill={rgb, 255:red, 155; green, 155; blue, 155 }  ,fill opacity=0.6 ] (125,130) -- (169.5,130) -- (169.5,168) -- (125,168) -- cycle ;
\draw    (170.5,149.5) -- (183,149.5) ;
\draw  [fill={rgb, 255:red, 155; green, 155; blue, 155 }  ,fill opacity=0.6 ] (305,74) -- (349.5,74) -- (349.5,112) -- (305,112) -- cycle ;
\draw    (348.5,93.95) -- (366,93.95) ;
\draw  [fill={rgb, 255:red, 155; green, 155; blue, 155 }  ,fill opacity=0 ] (246,74) -- (290.5,74) -- (290.5,112) -- (246,112) -- cycle ;
\draw    (291.5,93.5) -- (304,93.5) ;
\draw  [fill={rgb, 255:red, 155; green, 155; blue, 155 }  ,fill opacity=0.6 ] (305,130) -- (349.5,130) -- (349.5,168) -- (305,168) -- cycle ;
\draw    (350.5,149.95) -- (368,149.95) ;
\draw  [fill={rgb, 255:red, 155; green, 155; blue, 155 }  ,fill opacity=0.6 ] (246,130) -- (290.5,130) -- (290.5,168) -- (246,168) -- cycle ;
\draw    (291.5,149.5) -- (304,149.5) ;
\draw  [fill={rgb, 255:red, 155; green, 155; blue, 155 }  ,fill opacity=0 ] (427,74) -- (471.5,74) -- (471.5,112) -- (427,112) -- cycle ;
\draw    (472.5,93.95) -- (490,93.95) ;
\draw  [fill={rgb, 255:red, 155; green, 155; blue, 155 }  ,fill opacity=0 ] (368,74) -- (412.5,74) -- (412.5,112) -- (368,112) -- cycle ;
\draw    (413.5,93.5) -- (426,93.5) ;
\draw  [fill={rgb, 255:red, 155; green, 155; blue, 155 }  ,fill opacity=0 ] (549,75) -- (593.5,75) -- (593.5,113) -- (549,113) -- cycle ;
\draw    (592.5,94.95) -- (610,94.95) ;
\draw    (535.5,94.5) -- (548,94.5) ;
\draw  [fill={rgb, 255:red, 155; green, 155; blue, 155 }  ,fill opacity=0.6 ] (428,130) -- (472.5,130) -- (472.5,168) -- (428,168) -- cycle ;
\draw    (473.5,149.95) -- (491,149.95) ;
\draw  [fill={rgb, 255:red, 155; green, 155; blue, 155 }  ,fill opacity=0 ] (369,130) -- (413.5,130) -- (413.5,168) -- (369,168) -- cycle ;
\draw    (414.5,149.5) -- (427,149.5) ;
\draw  [fill={rgb, 255:red, 155; green, 155; blue, 155 }  ,fill opacity=0.6 ] (550,131) -- (594.5,131) -- (594.5,169) -- (550,169) -- cycle ;
\draw    (593.5,150.95) -- (611,150.95) ;
\draw  [fill={rgb, 255:red, 155; green, 155; blue, 155 }  ,fill opacity=0.6 ] (491,131) -- (535.5,131) -- (535.5,169) -- (491,169) -- cycle ;
\draw    (536.5,150.5) -- (549,150.5) ;

\draw (653,197) node [anchor=north west][inner sep=0.75pt]   [align=left] {time};
\draw (139,193.4) node [anchor=north west][inner sep=0.75pt]    {$t$};
\draw (190,193.35) node [anchor=north west][inner sep=0.75pt]    {$t+1$};
\draw (544,14) node [anchor=north west][inner sep=0.75pt]   [align=left] {Original chain};
\draw (625,165) node [anchor=north west][inner sep=0.75pt]   [align=left] {Fork};
\draw (14,83.4) node [anchor=north west][inner sep=0.75pt]    {$B_{0}$};
\draw (69,83.4) node [anchor=north west][inner sep=0.75pt]    {$B_{t-1}$};
\draw (136,83.4) node [anchor=north west][inner sep=0.75pt]    {$B_{t}$};
\draw (137,140.4) node [anchor=north west][inner sep=0.75pt]    {$B_{t}^{A}$};
\draw (637,142) node [anchor=north west][inner sep=0.75pt]   [align=left] {6 blocks};
\draw (635,87) node [anchor=north west][inner sep=0.75pt]   [align=left] {4 blocks};
\draw (254,194.4) node [anchor=north west][inner sep=0.75pt]    {$t+2$};
\draw (312.5,194.35) node [anchor=north west][inner sep=0.75pt]    {$t+3$};
\draw (376,193.35) node [anchor=north west][inner sep=0.75pt]    {$t+4$};
\draw (437.5,194.4) node [anchor=north west][inner sep=0.75pt]    {$t+5$};
\draw (498.5,194.35) node [anchor=north west][inner sep=0.75pt]    {$t+6$};
\draw (554.5,194.35) node [anchor=north west][inner sep=0.75pt]    {$t+7$};

\end{tikzpicture}
}
\caption{\textbf{Majority Attack under Proof of Stake}}
\label{fig:PoSattack}
\end{figure}
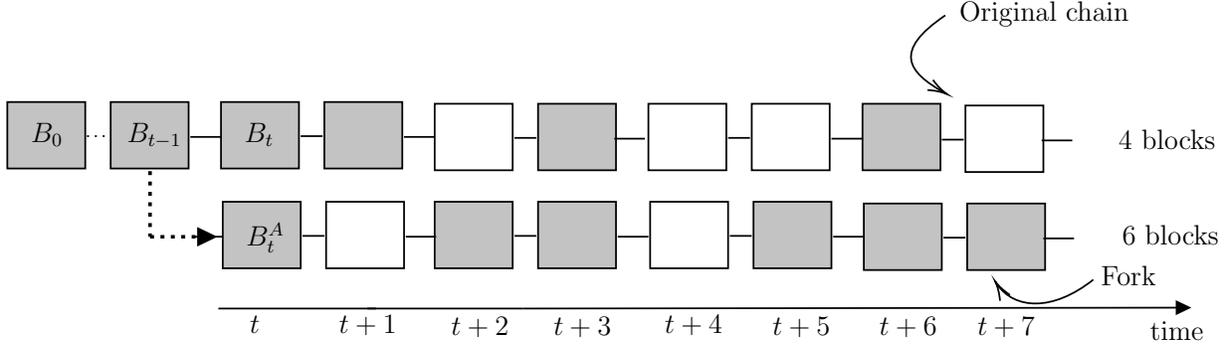

Suppose that an entity that controls majority of the stake, $s_A>\frac{1}{2} S$ creates a fork. 
Under PoS, a validator cannot attach a block to the blockchain unless they are ``called by the protocol'' to do it. Let $B_{t-1}$ be the last block before the fork. To create the fork, $A$ accepts $B_{t-1}$, but thereafter does not accept any block created by other validators than themself. If the protocol running on his local machine calls for a validator other than $A$, there is no response, and after $\tau_s$ the protocol calls another validator. Since the stake in the protocol has not changed, $A$ is called by the protocol with probability $\frac{s_A}{S}$ every interval $\tau_s$, and only these blocks are added to the attacking branch.

At the same time, $A$ withholds their blocks from the honest branch. Thus, during the attack, a block is added to the honest branch with probability $1-\frac{s_A}{S}$ every interval $\tau_s$. Note that with $s_A>\frac{1}{2}S$, the attack will succeed with probability 1 at some point. Let $\mathcal{L}=l\,\tau_s$ be the expected length of the attack in clock time. During $\mathcal{L}$, attacker $A$ adds in expectation $l\, \frac{s_A}{S}$ blocks, and these are all the blocks on the attacking blockchain. Without the attack, over $\mathcal{L}$, $A$ would also add $l\, \frac{s_A}{S}$ blocks in expectation. Thus, $A$ receives proposes the same number of blocks with and without the attack.

If $V_{attack}$ is the net present discounted expected payoff to $A$ from a successful attack, then the incentive compatibility constraint is:
$$k\frac{s_A}{S}(R_s+\Phi_s)-k\,\tau_s\, r\, s_A e_s \ge V_{attack} - C_{attack}$$
Note that, by free entry, the left-hand side of this inequality is $0$ so the constraint collapses to $C_{attack} \ge V_{attack}$. Here, $C_{attack}$ is:
$$C_{attack}= r\, s_A e_s \mathcal{L} - \frac{s_A}{S}(\tilde{R}_s+\tilde{\Phi}_s)$$
where $\tilde{R}_s$ and $\tilde{\Phi}_s$ are the block rewards and transaction fees earned during an attack.

Note that, as in Proof of Work, it is likely that $\tilde{R}_s = R_s$. In addition, using the free entry condition, $C_{attack}$ becomes $\frac{s_A}{\tau_s S}\big(\Phi_s - \tilde{\Phi}_s\big)\mathcal{L}$. Thus, the incentive compatibility condition becomes:
$$\frac{s_A}{\tau_S S}\big(\Phi_s - \tilde{\Phi}_s\big)\mathcal{L} \ge V_{attack}$$
By the same argument as for Proof of Work (that is, Lemma~\ref{lem:higher.fees}), it is easy to see that $\tilde{\Phi}_s > \Phi_s$. Thus, we have demonstrated the following result:
\begin{proposition}
    An equilibrium without a majority attack does not exist if there exists a majority validator even if $V_{attack} = 0$.
\end{proposition}
Note that this is a stronger result than under Proof of Work due to the symmetric and linear costs associated with being a validator, and fully inside attack.

\section{Conclusion}

In designing the Bitcoin network, \cite{nakamoto2008bitcoin} concluded that "[a]ny needed rules and incentives can be enforced with this [Proof of Work] consensus mechanism." While it was thought that an entity that controlled the majority of miners could attack the blockchain, the assumption was that this would be costly. This paper has demonstrated that this contention does not hold. Using within-protocol mechanisms only but for the special case of a purely external attacker, an attack can be carried out with negative costs. Thus, what regulates the security of the network is purely external. 

There are many candidates for external regulation, but these are often contingent on the precise motives for an attack. For example, a double-spend attack by which an attacker censors past transactions in order to re-spend tokens may require an agreed-upon escrow period (something \cite{nakamoto2008bitcoin} suggested). However, even here, a negative cost attack could be carried out indefinitely. In other cases where the attacker remains invested in the network post-attack (say, directly or indirectly because of the ownership of specialized processors), an anticipated fall in the value of crypto-tokens could deter an attack. Again, how this would operate is external to the protocol and, at present, the evidence that it is a regulating device is mixed.\footnote{See \cite{kwon2019bitcoin}, \cite{ramos2021great} and \cite{shanaev2019cryptocurrency}.} This points to future research focussing on the nature of $V_{attack}$ and developing an understanding of the external institutional and social mechanisms that may reduce or eliminate it.

\appendix 
\section*{Appendix}

\begin{lemma}\label{lem:equal.mining.rewards}
    For non-linear costs, if $\tilde{R}+\tilde{\Phi}=R+\Phi$, the net cost of an attack is zero when $\tilde{h}_A\geq \underline{h}_A$, and strictly positive otherwise.
\end{lemma}

\begin{proof}
    First note that $\max_{h_A} \{\frac{h_A}{D}(R+\Phi)-c_A(h_A)\}=-\min_{h_A}\{c_A(h_A)-\frac{h_A}{D}(R+\Phi)\}$. Moreover, under  $\tilde{R}+\tilde{\Phi}=R+\Phi$, $\min_{h_A}\{c_A(h_A)-\frac{h_A}{D}(R+\Phi)\}=\min_{h_A}\{c_A(h_A)-\frac{h_A}{D}(\tilde{R}+\tilde{\Phi})\}$. If $\tilde{h}_A\geq \underline{h}_A$, the net cost of an attack, i.e., the left hand side of~\eqref{IC2} becomes
    $$\underbrace{\left(\max_{h_A} \{\frac{h_A}{D}(R+\Phi)-c_A(h_A)\}+ \min_{h_A}\{c_A(h_A)-\frac{h_A}{D}(\tilde{R}+\tilde{\Phi})\}\right)}_{=0}\mathcal{L}=0\, .$$

    When $ \underline{h}_A$ is binding, i.e.  $\tilde{h}_A<\underline{h}_A$, then $C_{attack}>\min_{h_A}\{c_A(h_A)-\frac{h_A}{D}(\tilde{R}+\tilde{\Phi})\}$, and thus the net cost of an attack is
    \begin{small}
    
    $$\hspace*{-1cm} \left(\max_{h_A} \{\frac{h_A}{D}(R+\Phi)-c_A(h_A)\}+C_{attack}\right)\mathcal{L}>\left(\max_{h_A} \{\frac{h_A}{D}(R+\Phi)-c_A(h_A)\}+ \min_{h_A}\{c_A(h_A)-\frac{h_A}{D}(\tilde{R}+\tilde{\Phi})\}\right)\mathcal{L}=0 \, .$$
    
    \end{small}
\end{proof}

\begin{lemma}\label{lem:smaller.mining.rewards}
    When $\tilde{R}+\tilde{\Phi}<R+\Phi$, the net cost of the attack is always strictly positive.
\end{lemma}
\begin{proof}
    For the linear cost, see Section~\ref{sec:linear}.

    For the non-linear case, note that 
    \begin{multline*}
        \min_{h_A}\{c_A(h_A)-\frac{h_A}{D}(\tilde{R}+\tilde{\Phi})\}
        \, =\, c_A(\tilde{h}_A)-\frac{\tilde{h}_A}{D}(\tilde{R}+\tilde{\Phi})
        \, >\\>\, c_A(\tilde{h}_A)-\frac{\tilde{h}_A}{D}(R+\Phi) \,
        \geq \, c_A(h^*_A)-\frac{h^*_A}{D}(R+\Phi)\, 
        =\, \min_{h_A}\{c_A(h_A)-\frac{h_A}{D}(R+\Phi)\}
    \end{multline*}
    Then the net cost of attack is
    $$\Bigg(\underbrace{\max_{h_A} \{\frac{h_A}{D}(R+\Phi)-c_A(h_A)\}}_{=-\min_{h_A}\{c_A(h_A)-\frac{h_A}{D}(R+\Phi)\}}+\underbrace{C_{attack}}_{\geq \min_{h_A}\{c_A(h_A)-\frac{h_A}{D}(\tilde{R}+\tilde{\Phi})\}}\Bigg)\mathcal{L} >0 $$
\end{proof}

\newpage

\typeout{}
\bibliography{main}
\bibliographystyle{apalike}

\end{document}